\documentclass[journal]{IEEEtran}

\usepackage{amsmath} 
\usepackage{amssymb} 
\usepackage{amsthm}
\usepackage{cite}
\usepackage{multicol} 
\usepackage{listings}
\usepackage{dsfont} 
\usepackage{graphicx}
\usepackage{color}
\usepackage{ifpdf}

\usepackage{subfigure}
\usepackage{setspace} 
\usepackage{nicefrac} 
\usepackage{units} 
\DeclareGraphicsExtensions{.eps}
\newcommand{\f}[1]{\ensuremath{\mathbf{#1}}}                                        % bold math
\newcommand{\T}[1]{\text{#1}}                                                       % text in math

\newcommand{\figref}[1]{Fig. \ref{#1}}                                              % Figure reference

\begin{document}

\title{Robust Rate Adaptation and Proportional Fair Scheduling with Imperfect CSI}

\author{Richard Fritzsche, Peter Rost, and Gerhard P. Fettweis}% <-this % stops a space
%\thanks{R. Fritzsche is with the Department
%of Electrical and Computer Engineering, Technische Universit\"at Dresden, Vodafone Chair Mobile Communications Systems, Dresden, Germany.}% <-this % stops a space
%\thanks{P. Rost is with NEC Heidelberg, Germany.}% <-this % stops a space
%\thanks{G. P. Fettweis is with the Department
%of Electrical and Computer Engineering, Technische Universit\"at Dresden, Vodafone Chair Mobile Communications Systems, Dresden, Germany.}% <-this % stops a space
%\thanks{Manuscript submitted July 30, 2014}}

\maketitle

\begin{abstract}
In wireless fading channels, multi-user scheduling has the potential to boost the spectral efficiency by exploiting diversity gains. In this regard, proportional fair (PF) scheduling provides a solution for increasing the users' quality of experience by finding a balance between system throughput maximization and user fairness. For this purpose, precise instantaneous channel state information (CSI) needs to be available at the transmitter side to perform rate adaptation and scheduling. However, in practical setups, CSI is impaired by, e.g., channel estimation errors, quantization and feedback delays. Especially in centralized cloud based communication systems, where main parts of the lower layer processing is shifted to a central entity, high backhaul latency can cause substantial CSI imperfections, resulting in significant performance degradations. In this work robust rate adaptation as well as robust PF scheduling are presented, which account for CSI impairments. The proposed rate adaptation solution guarantees a fixed target outage probability, which is of interest for delay critical and data intensive applications, such as, video conference systems. In addition to CSI imperfections the proposed scheduler is able to account for delayed decoding acknowledgements from the receiver. 
\end{abstract}

%\IEEEpeerreviewmaketitle 

\section{Introduction}
\label{sec:introduction}
\IEEEPARstart{T}{he} steadily increasing throughput demand of mobile users together with the limited radio spectrum in cellular communication systems requires an efficient resource allocation. In orthogonal frequency division multiple access (OFDMA) where multiple users are served on different non-overlapping time-frequency resources, modern scheduling algorithms exploit multi-user diversity by accounting for the individual time variations of the users' mobile radio channel. 
\par
In order to perform scheduling, instantaneous channel information needs to be available at the transmitter side. However, channel state information (CSI) is typically impaired due to a variety of sources. In frequency division duplex (FDD) systems, CSI imperfections result from noisy pilot reception at the receiver and erroneous feedback transmission due to quantization errors as well as delays between channel observation and transmission \cite{LHL08, FOF13b}. The latter aspect is especially relevant at slow fading channels \cite{TV08}, where the provided CSI is still fairly correlated with the actual channel. Due to the lack of exact channel knowledge available at the transmitter, the rate achievable is not known precisely. This causes a non-optimal scheduling decision as well as imperfect rate assignment. If the channel amplitude is overestimated, decoding errors may occur at the receiver, which lead to outages of the respective data blocks and ultimately to a degradation in spectral efficiency \cite{TV08}.
\subsection{Related Literature}
In the area of multi-user scheduling, different objectives have been investigated over the last two decades. Opportunistic scheduling maximizes the system throughput by selecting the user with the best channel \cite{KH95, LG01}. However, with this approach the system may be unfair towards users with poor channel conditions, which might never be scheduled. User fairness can be achieved by employing throughput balancing or round robin scheduling \cite{LBS99} which comes at the cost of a significant degradation in the overall throughput. A compromise between throughput maximization and fairness is of interest \cite{Kel97} and has been reflected by the proportional fair (PF) scheduler, which maximizes the sum of the logarithmic user-throughputs \cite{JPP00, BW01}. PF multi-user scheduling has been investigated for several scenarios, such as, single-cell \cite{CB07} and multi-cell networks \cite{CC06, ZFL11}. An extension to multi antenna systems has been discussed in \cite{PSL03, Lau05}. A PF algorithm which accounts for intra-cell as well as inter-cell interference has been proposed by \cite{LNZ10}.
\par
The publications mentioned above inherently assume that CSI is precisely available at the scheduler. Scheduling with imperfect CSI has been studied in several works considering throughput maximization \cite{WE09, AAK+11, Ros12}. A PF scheduling scheme, which provides robustness against imperfect channel knowledge has been presented in our preliminary work \cite{FRF14}.
\par
\begin{figure}[t]
    \centering
        \includegraphics[width=40mm]{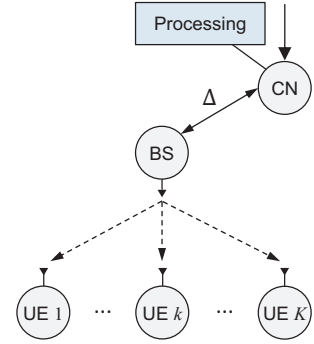}
    \caption{Setup with $K$ user equipments (UEs) served by a base station (BS), which is connected to a central node (CN). According to a cloud-based architecture, the main PHY and MAC layer processing is performed at a centralized data center.}
    \label{fig:setup}
\end{figure}
For delay constraint applications, the system operates at a fixed outage probability and hybrid automatic repeat request (HARQ) \cite{SK76, LMS07} is subject to a maximum number of re-transmissions in order to guarantee a certain quality of service (QoS). For applications with strict conditions on latency requirements (e.g., video conference systems, tactile internet), HARQ may not be an appropriate mechanism. Moreover, round trip delays may exceed respond time limitations, required by HARQ protocols (e.g. LTE-A). In LTE-A the transmission rate can be adapted by means of outer loop link adaptation (OLLA), which is based on acknowledgments of previous transmissions. However, in scenarios where the round trip delay is in the range of the channels coherence time, also the acknowledgment is insufficiently correlated with the actual channel.
In this work, we restrict our-selves to inner loop link adaptation (ILLA), where outage do not impact the rate adaptation of following transmissions.
\par
For generalized fading channels, outage probability has been studied in \cite{KAS00} for receive combining techniques. Rate assignment for throughput maximization w.r.t. slow fading channels has been presented in \cite{KCB+11}, while fast fading channels have been considered in \cite{WJ10} employing a fixed outage probability. The effect of imperfect CSI on the scheduling performance has been investigated in \cite{GM13}. In \cite{GM13}, the authors presented an analytical framework for jointly analyzing the performance of certain state of the art methods regarding feedback transmission, rate adaptation and scheduling. In contrast to our work the algorithms itself are not optimized. Instead the framework allows to select the optimal combination of particular methods for a given scenario.
\par
Recently, cloud based architectures in cellular communication systems attract significant attention \cite{RBD+14}. The basic idea is the outsourcing of computationally complex algorithms of the physical (PHY) and the medium access control (MAC) layer to a data center which provides processing power as well as high scalability. In such an architecture, the scheduler is located at a central node (CN), which is accessible to a large number of BSs via the backhaul network. In practice, the backhaul connections cause technology-dependent delays. Such extra latency causes outdated CSI as well as delayed transmission acknowledgements.

\subsection{Contribution of this Work}
In this paper, we study the effect of imperfect CSI in the context of a cloud-based cellular communication system, where rate adaptation and scheduling is performed at a centrally located data center. In order to provide large user throughput together with a guaranteed QoS, a robust rate adaptation method is presented, which guarantees a fixed outage probability. The improved rate adaptation is incorporated into the PF scheduling metric in order to achieve robustness against impaired channel knowledge. The algorithm additionally accounts for delayed acknowledgements of successful transmissions. 
\par
This paper is an extension of our previous work \cite{FRF14}. It provides deeper
insight on the derivations of proposed methods regarding rate adaptation and PF scheduling. Moreover, algorithmic solutions w.r.t. a practical implementation as well as an extensive evaluation part is given. While the impact of backhaul delays has mainly been addressed in the context of cooperative cellular systems \cite{FF11, DF11}, this work focuses on non-cooperative cloud-based communication systems.

\par
The remainder of this paper is structured as follows. The system model is presented in Section \ref{sec:systemModel} before outage probability results with imperfect CSI are given in Section \ref{sec:outageProbability}. Section \ref{sec:proportionalFairScheduling} examines the proposed PF scheduling scheme, while
Section \ref{sec:simulationResults} shows simulation results followed by conclusions in Section \ref{sec:conclusions}.
\par

\subsection{Notation} Conjugate, transposition and conjugate
transposition are denoted by $(\cdot)^*$, $(\cdot)^T$ and $(\cdot)^H$, respectively.
%The trace of a matrix is $\text{tr}(\cdot)$, $\text{det}(\cdot)$ denotes
%determinant while $||\cdot||$ is used for Frobenius norm. $\text{dg}(\cdot)$ replaces each off diagonal matrix element with
%zero, $\text{vec}(\cdot)$ stacks all matrix columns into a vector and $\odot$ refers to element wise multiplication. 
Expectation is $\mathds{E}\{\cdot\}$. The probability of an event $A$ is $\mathds{P}\{A\}$, while the probability of event $A$ conditioned on a given event $B$ is $\mathds{P}\{A | B\}$. Furthermore, $\mathds{C}$ denotes the set of
complex numbers and $\mathcal{N}_{\mathds{C}}(\mu,\sigma^2)$ refers to a
complex normal distribution with mean $\mu$ and
variance $\sigma^2$.

\section{System Model}
\label{sec:systemModel}
This work considers a single-cell downlink system comprising $K$ user equipments (UEs), which are served by the base station (BS) on different non-overlapping resource elements, as known from OFDMA. The BS is connected to a central node (CN), e.\,g., a data center with sufficient computational power and scalability. Accordingly, the main PHY and MAC layer algorithms are performed in a centralized fashion at the CN (see \figref{fig:setup}). In this work, the algorithms particularly comprise rate adaptation and multi-user scheduling. 
\par
Centralized processing requires that instantaneous channel information is available at the CN. Assuming the employment of an FDD system, the downlink channel cannot be obtained by observing the uplink channel. Hence, CSI needs to be provided to the BS, from where it is then forwarded to the CN. The whole information flow including signaling as well as data transmission is illustrated as message sequence chart in \figref{fig:sequence}. 
%\par 
%The remainder of this section details the model for downlink data transmission, the definition of outage and throughput in case of imperfect CSI as well as the model for imperfect CSI itself.
\begin{figure}[t]
	\centering
	\includegraphics[width=80mm]{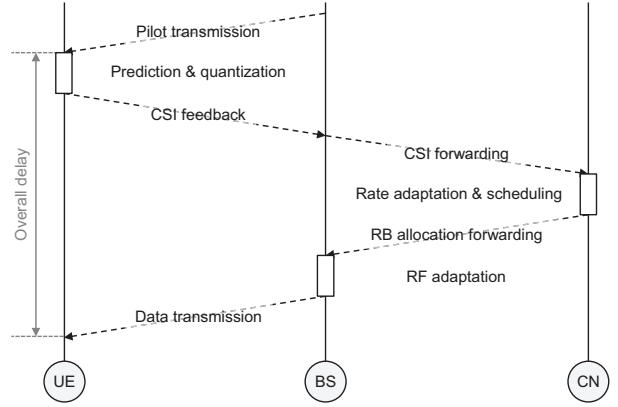}
	\caption{The message sequence chart illustrates the process between channel observation and transmission, including rate adaptation and scheduling at the central node (CN).}
	\label{fig:sequence}
\end{figure}
\subsection{Data Transmission}
The instantaneous radio channel at time slot $n$ between the BS and UE $k$ is denoted with $h_k[n]\sim\mathcal{N}_{\mathds{C}}(0,\lambda_k)$, where $\lambda_k$ is the mean channel gain of UE $k$, reflecting path loss and shadow fading effects. A time slot reflects the duration of a transmission block, which refers to a collection of continuous symbol transmissions in time and frequency where the channel is assumed to be constant. The radio resources of each transmission block are scheduled only to a single user. Consequently, no inter-user interference needs to be considered.
\par
Assuming Gaussian distributed channel input $x_k\sim\mathcal{N}_{\mathds{C}}(0,\rho)$ with transmit power $\rho$, the received signal is given as
\begin{equation}
y_k = h_k[n]x_k+\eta_k,
\end{equation}
where $\eta_k\sim\mathcal{N}_{\mathds{C}}(0,1)$ is the zero mean complex Gaussian receiver noise at UE $k$ with power equal to one.
\par 
With precise knowledge of the instantaneous channel gain $|h_k[n]|^2$ at the transmitter as well as precise knowledge of the complex channel $h_k[n]$ (including amplitude and phase) at the receiver side, the capacity (maximum achievable transmission rate) for UE $k$ at time slot $n$ is given by
\begin{equation}
C_k[n] = \log_2\left(1 + \rho |h_k[n]|^2\right).
\label{equ:rate_perfect}
\end{equation}
Considering a long-term power constraint, the capacity is achieved by adapting the transmit power $\rho$ at each time slot to the respective channel state \cite{TV08}. However, this work assumes an instantaneous power constraint with a fixed power allocation $\rho$ at each time slot. The signal-to-noise-ratio (SNR) experienced at UE $k$ is then given by $\gamma_k = \rho\lambda_k$. 
\par 

\subsection{Outage Probability and Throughput}
The channel capacity in (\ref{equ:rate_perfect}) is only achievable with precise rate allocation, i.e., the maximum rate supported by the instantaneous channel needs to be assigned for transmission. However, with imperfect channel knowledge, the achievable rate is not known at the transmitter. Consequently, the allocated transmission rate might not match the channel capacity. 
\par
If the rate allocated is below capacity, the channel's potential is not fully exploited. In this case only a reduced spectral efficiency is achieved. On the other hand, if the rate assigned exceeds capacity, outage of the transmitted data block occurs, i.e, the codeword received at the UE cannot be decoded correctly \cite{TV08}.
\par
Let function $S_k[n]$ indicate the success of transmission to user $k$ in time slot $n$, i.e., $S_k[n]$ is equal to one if the codeword has been successfully decoded and zero otherwise. With the rate assigned for transmission $\bar{R}_k[n]$, the actually experienced rate at UE $k$ in time slot $n$ is $R_k[n] = S_k[n]\bar{R}_k[n]$. The probability of outage is given by
\begin{equation}
\begin{split}
p_{\T{out},k}[n] &= \mathds{P}\left\{S_k[n] = 0\right\}\\
& = \mathds{P}\left\{\log_2\left(1+\rho |h_k[n]|^2\right)<\bar{R}_k[n]\right\}.
\end{split}
\label{equ:outage_probability}
\end{equation}
From the transmitter point of view, $h_k[n]$ is a random variable. An extension of (\ref{equ:outage_probability}) to the case with side information is stated in Section \ref{sec:outageProbability}.
\par
The throughput of UE $k$ obtained until time slot $N$ is given by
\begin{equation}
T_k[N] = \frac{1}{N}\sum_{n=1}^NR_k[n] = \frac{1}{N}\sum_{n=1}^NS_k[n]\bar{R}_k[n].
\label{equ:net_rate}
\end{equation}
The throughput expression in (\ref{equ:net_rate}) refers to the average rate of user $k$,  experienced up to time slot $N$. 

\subsection{Imperfect CSI}
\label{sec:imperfectCSI}
In this section, we present the model for imperfect CSI used in this work. In FDD systems, the downlink channel cannot be obtained from corresponding uplink measurements but is observed at the UE side and provided to the BS via an uplink control channel with limited capacity. Consequently, the CSI available at the BS is impaired by channel estimation errors, feedback quantization and potential delays between channel estimation and actual data transmission. 
In this work, we use the model derived in \cite{FOF13} which considers all three effects. CSI is obtained by minimum mean square error (MMSE) estimation/prediction, i.\,e., the unknown channel 
\begin{equation}
h_k[n]\sim\mathcal{N}_{\mathds{C}}\left(\hat{h}_k[n], \epsilon_k\right)
\end{equation}
can be represented by a complex Gaussian random variable, whose mean value refers to the currently available channel estimate $\hat{h}_k[n]$, while its variance $\epsilon_k$ reflects the respective channel uncertainty \cite{Kay93}. Note that $\epsilon_k/\lambda_k$ refers to the normalized CSI impairment and takes values in the interval $[0; 1]$. The unknown channel can equivalently be written as the sum of the channel estimate and a zero mean Gaussian random error $e_k[n]\sim\mathcal{N}_{\mathds{C}}(0,\epsilon_k)$, i.\,e.
\begin{equation}
h_k[n] = \hat{h}_k[n] + e_k[n].
\label{equ:channel_imperfect_CSI}
\end{equation}
Note that in this model, the mean gain of the channel estimate $\mathds{E}\{|\hat{h}_k[n]|^2\} = \lambda_k(1-\epsilon_k/\lambda_k)$ is reduced according to the error variance to $\epsilon_k$. 
According to \cite{FOF13}, $\epsilon_k$ can be calculated based on the pilot SNR, the quantization resolution, the delay and the processing window length $W$ of the prediction filter. In this regard the time variation of the channel is reflected by the correlation between two channel coefficients, delayed by $\Delta$ time slots:
\begin{equation}
c[\Delta] = \mathds{E}_{n}\{h_k[n]h_k^*[n+\Delta]\}/\lambda_k = J_0\left(q\frac{\Delta}{T_{\T{c}}}\right),
\label{equ:bessel}
\end{equation} 
where $T_{\T{c}}$ and $J_0$ are the normalized coherence time (normalized to the time slot duration) and the zero-th order Bessel function of the first kind, respectively. The constant $q=J_0^{-1}(0.5)$ ensures that $T_{\T{c}}$ is referring to the $50\%$ coherence time. In this work, slow fading is assumed such that the channels of subsequent time slots show strong correlations. Using (\ref{equ:bessel}),  the covariance vector 
\begin{equation}
\f{c}[\Delta]=[c[\Delta],\dots c[\Delta+W-1]]^T
\end{equation} 
as well as the covariance matrix $\f{C}=[\f{c}[0],...,\f{c}[-\Delta]]$ are defined, considering that $W$ subsequent channel observations are used for prediction.
\par
Denoting the number of pilot signals within a transmission block as $N_P$ and the number of feedback bits per complex channel coefficient as $Q$, the error variance is given by
\begin{equation}
\epsilon_k = (1-2^{-Q})\f{c}^H[\Delta](\f{C}+\rho N_P)^{-1}\f{c}[\Delta].
\end{equation}
A more detailed derivation and explanation of the feedback model is given in \cite{FOF13}.

\section{Robust Rate Adaptation with Imperfect CSI}
\label{sec:outageProbability}
In this section, the proposed rate adaptation scheme is derived.
The algorithm is performed for each user individually. Since this work focuses on ILLA, the rate adaptation at a particular time-slot is not affected by previous transmissions. Independence between multiple users and subsequent time-slots allows us to omit the UE index $k$ as well as the time slot index $n$ for improving readability in this section. Note, that the impact of feedback latency is still fully reflected by the error variance $\epsilon$.
Furthermore, the amplitude of the actual channel and the CSI are denoted as $g=|h|$ and $\hat{g}=|\hat{h}|$, respectively.
\par
Utilizing the results of \cite{KAS00, TV08}, the probability of outage conditioned on the available channel information $\hat{g}$ is given as
\begin{equation}
\renewcommand\arraystretch{1.8}
\begin{array}{rl}
p_{\T{out}} &= \mathds{P}\left\{\log_2\left(1+\rho g^2\right)<\bar{R} \quad|\quad \hat{g}\right\}\\
&=\mathds{P}\left\{g<\sqrt{\left(2^{\bar{R}}-1\right)/\rho} \quad|\quad \hat{g}\right\}\\
&=\mathcal{F}_{g|\hat{g}}\left(\sqrt{\left(2^{\bar{R}}-1\right)/\rho}\right),
\end{array}
\label{equ:outageProb}
\end{equation}
where $\mathcal{F}_{g|\hat{g}}(x)$ denotes the cumulative distribution function (CDF) of $g$ conditioned on side information $\hat{g}$. Eq. (\ref{equ:outageProb}) allows to obtain the outage probability for the given channel estimate $\hat{g}$ when rate $\bar{R}$ is allocated. However, in case of delay constrained applications, the outage probability is given as a constraint and the transmitter is interested in the rate allocated to satisfy this constrained. Consequently, the inverse function $\bar{R}(\hat{g}, p_{\T{out}})$ needs to be found.
\par
The amplitude of a complex Gaussian non-zero mean random variable follows a Rician distribution \cite{TV08}. Consequently, the probability density function (pdf) of $g$ conditioned on $\hat{g}$ results in
\begin{equation}
f_{g |\hat{g}}(g) = \frac{2g}{\epsilon}\T{exp}\left(-\frac{g^2+\hat{g}^2}{\epsilon} \right)\bar{J}_0\left(\frac{2g\hat{g}}{\epsilon} \right),
\label{equ:pdf_channel}
\end{equation}
where $\bar{J}_0(x)=\sum_{l=0}^\infty(x/2)^{2l}/(l!\Gamma(l+1))$ refers to the modified Bessel function of the first kind and order zero, while $\Gamma(x)=\int_{0}^{\infty}t^{x-1}e^{-t}dt$ is the gamma function.
\par
The CDF in (\ref{equ:outageProb}) is obtained by integrating (\ref{equ:pdf_channel}) over the interval $[0; b]$ with $b =\sqrt{\left(2^{\bar{R}}-1\right)/\rho}$:
\begin{equation}
\mathcal{F}_{g|\hat{g}}(b)= \int_0^b f_{g |\hat{g}}(g) dg.
\label{equ:cdf_channel1}
\end{equation}
With the results of \cite{Nut75}, the integral in (\ref{equ:cdf_channel1}) is given by
\begin{equation}
\renewcommand\arraystretch{1.8}
\begin{array}{rl}
\hspace{-0.2cm}\mathcal{F}_{g|\hat{g}}(b) \hspace{-0.2cm}&=1 - Q_1\left(\sqrt{\frac{2\hat{g}^2}{\epsilon}},\sqrt{\frac{2b^2}{\epsilon}} \right)\\
&=1 - \T{exp}\left(-\frac{\hat{g}^2 + b^2}{\epsilon}\right)\sum_{m=0}^\infty\left(\frac{\hat{g}}{b}\right)^m \bar{J}_m\left(\frac{2\hat{g}b}{\epsilon}\right) ,
\end{array}
\label{equ:cdf_channel2}
\end{equation}
where $Q_1$ is the Marcum Q-function and $\bar{J}_m(x)=\sum_{l=0}^{\infty}\frac{1}{l!\Gamma(l+m+1)}\left(\frac{x}{2}\right)^{2l+m}$ is the modified Bessel function of the first kind and order $m$. 
\par
Hence, the resulting outage probability with given $\hat{g}$ is given by
\begin{equation}
\renewcommand\arraystretch{2}
\begin{array}{rl}
p_{\T{out}}=1 - &\hspace{-0.3cm}\T{exp}\left(\frac{1-\rho\hat{g}^2 - 2^{\bar{R}}}{\rho\epsilon}\right)\cdot\\
&\sum_{m=0}^\infty\left(\frac{\hat{g}\sqrt{\rho}}{\sqrt{2^{\bar{R}}-1}}\right)^m \bar{J}_m\left(\frac{2\hat{g}\sqrt{2^{\bar{R}}-1}}{\epsilon\sqrt{\rho}}\right).
\end{array}
\label{equ:outage_last}
\end{equation}
However, (\ref{equ:outage_last}) cannot be reformulated in order to determine the outage rate $\bar R$ in closed form. Hence, numerical methods need to be applied to determine the $\bar R$ which guarantees a target outage probability $\bar{p}_{\T{out}}$.
\par
\begin{figure}[t]
	\centering
	\includegraphics[width=0.48\textwidth]{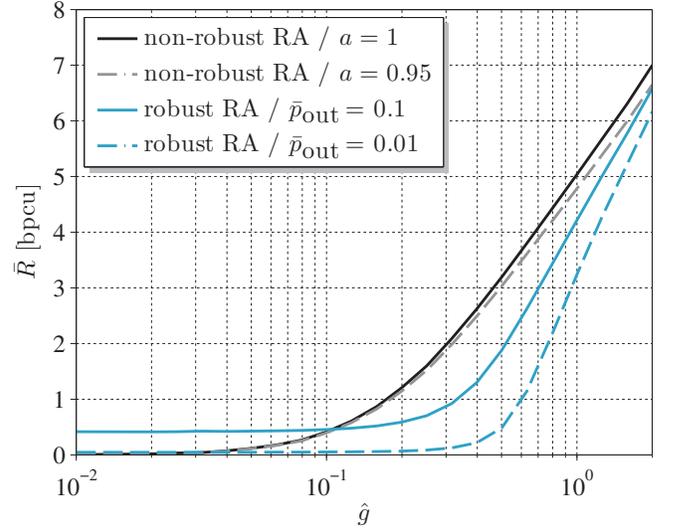}
	\caption{Rate assigned for transmission as a function of the amplitude of the available channel estimate.}
	\label{fig:assignedRate}
\end{figure}
\begin{figure}[t]
	\centering
	\includegraphics[width=0.48\textwidth]{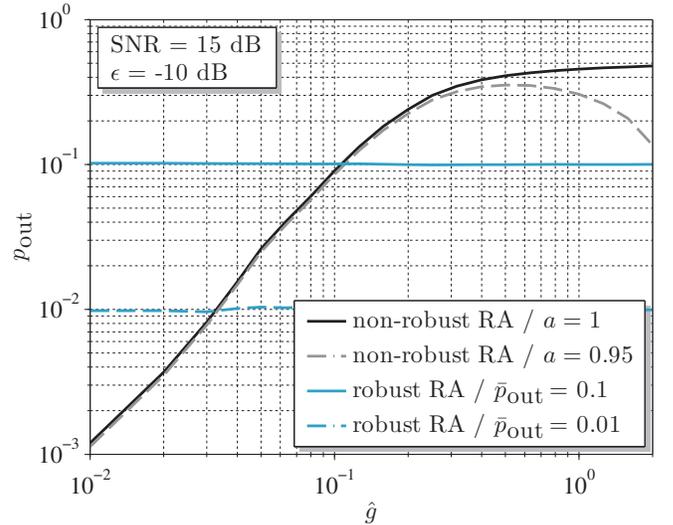}
	\caption{Resulting outage probability as a function of the amplitude of the available channel estimate.}
	\label{fig:resOutage}
\end{figure}
\figref{fig:assignedRate} shows the rate allocated as a function of the estimated channel amplitude $\hat g$. The blue curves refer to the proposed robust rate allocation (RA) scheme, which is given by the inverse function of (\ref{equ:outage_last}), as discussed above. The graphs are plotted for a target outage probability of $\bar{p}_{\T{out}}=0.1$ and $\bar{p}_{\T{out}}=0.01$. For comparison, non-robust RA is shown (black curve). The latter algorithm does not take into account that CSI is imperfect. Note, that with decreasing error variance $\epsilon$, the robust RA converges to the non-robust RA. The non-robust scheme can be parametrized by a back-off factor $a$, according to 
\begin{equation}
\bar{R} = a\log_2(1+\rho\hat{g}^2).
\label{equ:rateBackoff}
\end{equation}
Decreasing $a$ leads to a less aggressive rate allocation and a reduction of the outage probability. The results of the respective functions employing a back-off factor of $a=0.95$ are illustrated by the gray dashed line.
\par 
\begin{figure}[t]
	\centering
	\includegraphics[width=0.48\textwidth]{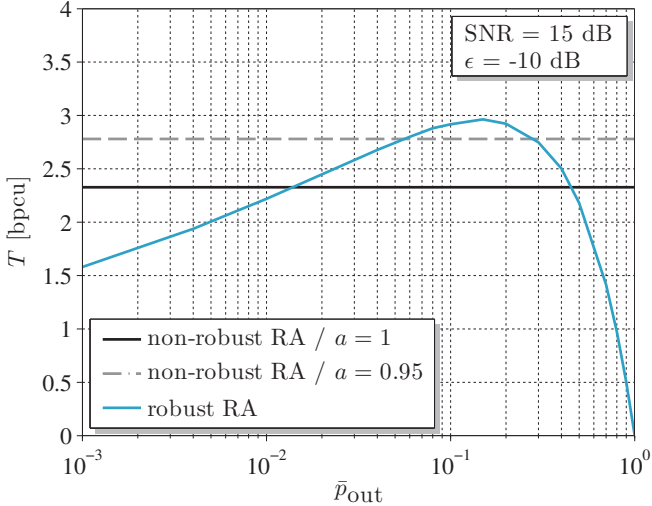}
	\caption{Mean user throughput as a function of the target outage probability.}
	\label{fig:netRate}
\end{figure}
\figref{fig:resOutage} shows that the robust scheme achieves the target outage probability for all channel amplitudes $\hat g$. Consequently the same QoS can be obtained for all transmissions. In contrast, non-robust RA selects the transmission rate based Shannon capacity. CSI is assumed to be precisely available. Hence, the target outage probability is ignored by the algorithm. As a consequence, the requested QoS constraints are only achieved incidentally for small amplitudes. However, reducing the back-off factor $\alpha$ can increase the percentage of fulfilled QoS constraints.
\par 
The throughput which results from the actually received transmission blocks is illustrated in \figref{fig:netRate}. In the low-$\bar{p}_{\T{out}}$ regime the guaranteed QoS comes with a reduction in throughput. For target outage probabilities between $0.1$ and $1$, the throughput decreases again. This effect results from a high packet error rate, which is forced by the robust algorithm in order to reach the target outage probability. However, this regime is of little practical relevance.

\section{Robust Proportional Fair Scheduling with Imperfect CSI}
\label{sec:proportionalFairScheduling}
In this section, the results derived for robust rate adaptation are applied to the proportional fair (PF) scheduler. It is shown that the robust PF solution is achieved by just substituting the rate assigned with the one resulting from robust rate adaptation. Consequently, accounting for imperfect CSI does not change the scheduler itself but only its input. 
\par
Furthermore, the optimal robust PF solution for delayed transmission acknowledgements is derived. In this case, the actual throughput during the last $\Delta$ transmissions is unknown to the scheduler as well. The optimal solution inherently estimates the unknown throughput by taking all possible cases into account.
\par 
The previous section presented rate adaptation independently of the actual user and time slot, i.\,e., the respective indexes $k$ and $n$ are utilized again. In the following, we first present the model for the expected rate, we then discuss the case of immediate acknowledgements, and finally we present the robust PF scheduler for delayed transmission acknowledgements. 

%===========================================================================================
\subsection{Expected Rate}
\label{sec:expectedRate}
%===========================================================================================

Assuming UE $k$ is served in time slot $n$, the expected rate is given by
\begin{equation}
\hat{R}_k[n] = \mathds{E}\{R_k[n]\} = (1 - p_{\T{out},k})\bar{R}[n], 
\end{equation}
which is the product of the probability of a successful transmission and the respective transmission rate. 
Since the experienced user rate is not precisely known to the transmitter, it refers to a random variable with mean value $\hat{R}_k[n]$. Consequently, the actual experienced rate is given by
\begin{equation}
R_k[n] = \hat{R}_k[n] + \Psi_k[n],
\label{equ:rateSplit}
\end{equation}
which is the sum of the expected rate and a bias term $\Psi_k[n]$. The experienced rate is either the rate $\bar{R}_k[n]$ assigned for transmission or zero. Hence, the bias term results in
\begin{equation}
\Psi_k[n] = 
\begin{cases}
    \bar{R}_k[n]-\hat{R}_k[n]&\text{if }\bar{R}_k[n]\leq C_k[n] \\
    -\hat{R}_k[n]&\text{otherwise}.\\
\end{cases}
\label{equ:biasTerm}
\end{equation}
By definition, the expectation of $\Psi_k[n]$ is zero, since
\begin{equation}
\renewcommand\arraystretch{1.8}
\begin{array}{rl}
\hspace{-0.37cm}\mathds{E}\{\Psi_k[n]\} \hspace{-0.2cm}&=(1\hspace{-0.07cm}-\hspace{-0.07cm}p_{\T{out},k})(\bar{R}_k[n]\hspace{-0.07cm}-\hspace{-0.07cm}\hat{R}_k[n])\hspace{-0.07cm}-\hspace{-0.07cm}p_{\T{out},k}\hat{R}_k[n]\\
&=(1\hspace{-0.07cm} - \hspace{-0.07cm}p_{\T{out},k})\bar{R}_k[n]\hspace{-0.07cm} -\hspace{-0.07cm} \hat{R}_k[n] =0.
\end{array}
\label{equ:expectationBiasTerm}
\end{equation}
With the given rate model the derivation of the PF scheduler can now be extended to the case of imperfect CSI, as it is stated in the following two sections.

%===========================================================================================
\subsection{Immediate Acknowledgements}
\label{sec:immediateCSI}
%===========================================================================================

In this section, we assume that the throughput of each user until the last time slot, $n-1$, is precisely known to the transmitter. In practice, this would require immediate transmission acknowledgements. Although this case is rather artificial, it is an important step towards the more practical case of delayed acknowledgements. In the following, we use the indicator function $I_k[n]$, which is equal to one if UE $k$ is scheduled at time slot $n$ and zero otherwise.
\newtheorem{Theorem}{Theorem}
\begin{Theorem}
The robust PF scheduler selects the user with the largest value
\begin{equation}
\hat{\nu}_k[n] = \mathds{E}\{\nu_k[n]\}=\frac{I_k[N]\hat{R}_k[N]}{T_k[N-1]}.
\label{equ:schedulingMetric_theorem}
\end{equation}
Consequently, the PF scheduler itself does not need to be changed. Only the rate assigned for transmission is adapted according to the robust algorithm stated in Section \ref{sec:outageProbability}.
\end{Theorem}
\begin{proof}
The throughput of UE $k$ experienced until time slot $N$ is given as
\begin{equation}
\renewcommand\arraystretch{1.8}
\begin{array}{rl}
T_k[N] &= \frac{1}{N}\sum_{n=1}^NI_k[n]R_k[n]\\
& = \frac{1}{N}\left(\sum_{n=1}^{N-1} I_k[n]R_k[n] + I_k[N]R_k[N]\right)\\
& = \frac{N-1}{N}T_k[N-1] + \frac{1}{N}I_k[N]R_k[N] = 0.\\
\end{array}
\label{equ:throughput_der}
\end{equation}
Inserting (\ref{equ:rateSplit}) into (\ref{equ:throughput_der}) results in the expected throughput, given as
\begin{equation}
\mathds{E}\{T_k[N]\} = \frac{N-1}{N}T_k[N-1] + \frac{1}{N}I_k[N]\hat{R}_k[N].
\label{equ:expectedThroughput}
\end{equation}
Note, that (\ref{equ:expectedThroughput}) results from the zero mean random bias term, as it is given in (\ref{equ:expectationBiasTerm}).
\par
The throughput values of all $K$ UEs given in (\ref{equ:throughput_der}) are collected in vector
\begin{equation}
\f{T}[N] = [T_1[N],\ldots,T_K[N]]^T.
\label{equ:vector}
\end{equation}
As defined in \cite{ZFL11}, the PF scheduler aims to maximize the utility function
\begin{equation}
u\left(\f{T}[N]\right) = \sum_{k=1}^K\log\left(T_k[N]\right)
\label{equ:utility}
\end{equation}
for an infinite number of time slots, expressed as
\begin{equation}
\max \underset{N\rightarrow\infty}{\lim}\sum_{k=1}^K\log\left(T_{k}[N]\right).
\end{equation}
The gradient of (\ref{equ:utility}) results in
\begin{equation}
\renewcommand\arraystretch{1.8}
\begin{split}
\nabla u\left(\f{T}[N]\right) &= \left[\frac{\partial u\left(\f{T}[N]\right)}{\partial T_1[N]},\ldots,\frac{\partial u\left(\f{T}[N]\right)}{\partial T_K[N]}\right]^T\\
&= \left[\left(T_1[N]\right)^{-1},\ldots,\left(T_K[N]\right)^{-1}\right]^T.\\
\end{split}
\label{equ:derivative}
\end{equation}
The second line comes from $\partial \log(x)/\partial x = 1/x$.
Since the scheduler acts on time slot basis, the utility increment is of interest, which can be expressed by
\begin{equation}
\renewcommand\arraystretch{2.1}
\begin{split}
\hspace{-0.22cm}& u\left(\f{T}[N]\right)-u\left(\f{T}[N-1]\right)\backsimeq\\  &\hspace{0.52cm}\backsimeq\nabla u\left(\f{T}[N-1]\right)^T\left(\f{T}[N]-\f{T}[N-1]\right)\\
&\hspace{0.52cm}=\sum_{k=1}^K\frac{T_k[N] - T_k[N-1]}{T_k[N-1]}\\
&\hspace{0.52cm}\overset{(\ref{equ:throughput_der})}{=}\sum_{k=1}^K\frac{I_k[N]R_k[N] - T_k[N-1]}{NT_k[N-1]}\\
&\hspace{0.52cm}=\sum_{k=1}^K\frac{I_k[N]\left(\hat{R}_k[N]+\Psi_k[N]\right)}{NT_k[N-1]}-\frac{K}{N}.\\
\end{split}
\label{equ:increment}
\end{equation}
The first and second line result from rearranging the difference quotient, while line three comes from inserting (\ref{equ:derivative}) and (\ref{equ:vector}). With the definition of throughput in (\ref{equ:throughput_der}) and the substitution of the actual transmission rate in (\ref{equ:rateSplit}) the last two lines can be stated. 
\par
Consequently, maximizing the utility increment (\ref{equ:increment}) results in scheduling the UE with the largest value for 
\begin{equation}
\nu_k[n] = \frac{I_k[N](\hat{R}_k[N]+\Psi_k[N])}{T_k[N-1]}.
\end{equation} 
However, $\Psi_k[N]$ is unknown to the scheduler and the expected utility increment $\mathds{E}\{f\left(\f{T}[N]\right) - f\left(\f{T}[N-1]\right)\}$ need to be maximized. As already stated in (\ref{equ:expectationBiasTerm}) the expectation $\mathds{E}\{\Psi_k[N]\}$ is equal to zero. Hence, scheduling the UE with the largest value for 
\begin{equation}
\hat{\nu}_k[n] = \mathds{E}\{\nu_k[n]\}=\frac{I_k[N]\hat{R}_k[N]}{T_k[N-1]},
\label{equ:schedulingMetric}
\end{equation}
maximizes the expected utility increment. Note, that with (\ref{equ:schedulingMetric}) also  the utility function $u(\mathds{E}\{\f{T}[N]\})$ is maximized. 
\end{proof}

%===========================================================================================
\subsection{Delayed Acknowledgements}
\label{sec:delayedCSI}
%===========================================================================================

In this subsection, the robust PF scheduler derived in Section \ref{sec:immediateCSI} is extended to the more realistic case of delayed transmission acknowledgements. One of the main sources of CSI imperfections is a delayed feedback and backhaul transmission. Hence, also the confirmation about the success of previous transmission is not known instantaneously to the scheduler. 

\begin{Theorem}
The robust PF scheduler for delayed transmission acknowledgements selects the user with the largest value 
\begin{equation}
\hat{\nu}_k[n] = \frac{I_k[N]\hat{R}_k[N]}{\tilde{T}_k[N]},
\label{equ:finalMetrictheorem}
\end{equation}
where $\tilde{T}_k[N] = 1/\sum_{m=1}^M\frac{\mathds{P}\{D=d_m\}}{w_0N+ d_m}$ denotes the throughput expected to be achieved at user $k$. In this regard the constant $w_0 = (N-\Delta-1)/N\cdot T_k[N-\Delta-1]$ is defined as well as the discrete random variable $D=\sum_{j=N-\Delta}^{N-1} I_k[j]R_k[j]$, which comprises $M = 2^\nu$ events, while $\nu = \sum_{n=N-\Delta}^{N-1}I_k[n]$).
\par
Consequently, beside the rate assigned for the upcoming transmission also the value for the throughput which is assumed to be achieved by user $k$ so far needs to be adapted.
\end{Theorem}
\begin{proof}
The throughput of UE $k$ achieved up to time slot $N$ is a random variable. 
First, it is known that UE $k$ has been scheduled for $\nu = \sum_{n=N-\Delta}^{N-1}I_k[n]$ transmissions within the period of interest. Hence, there are $M = 2^\nu$ possible combinations of successful and non-successful transmissions. Each combination has a certain probability, resulting from the probability of outage for each transmission.
The throughput equation (\ref{equ:throughput_der}) can be rewritten to
\begin{equation}
\renewcommand\arraystretch{1.8}
\begin{array} {rl}
\hspace{-0.22cm}T_k[N] \hspace{-0.2cm}&= \frac{1}{N}\big(\sum_{i=1}^{N-\Delta-1} I_k[i]R_k[i] + \\
&\quad+\sum_{j=N-\Delta}^{N-1} I_k[j]R_k[j] + I_k[N]R_k[N]\big)\\
& = \frac{N-\Delta-1}{N}T_k[N-\Delta-1] + \\
&\quad+\frac{1}{N}\sum_{j=N-\Delta}^{N-1} I_k[j]R_k[j] + \frac{1}{N}I_k[N]R_k[N].\\
\end{array}
\label{equ:througput_delay}
\end{equation}
While $w_0 = (N-\Delta-1)/N\cdot T_k[N-\Delta-1]$ is known to the scheduler, the remaining part of the sum in (\ref{equ:througput_delay}) is not, but can be split up according to (\ref{equ:rateSplit}). 
Therefore, the utility increment as stated in (\ref{equ:increment}) is given as 
\begin{equation}
\begin{split} 
\mathds{E}\{u\left(\f{T}[N]\right) \hspace{-0.02cm}&- u\left(\f{T}[N-1]\right)\} = \\
&=\frac{N-1}{N}\sum_{k=1}^K\frac{I_k[N]R_k[N]}{w_0N+ D}-\frac{K}{N},
\end{split}
\end{equation}
where $D=\sum_{j=N-\Delta}^{N-1} I_k[j]R_k[j]$ is a discrete random variable with $M = 2^\nu$ events, each of which denoted as $d_m$, $\forall m$.
Consequently, maximizing the expected utility increment lead to scheduling the user with the largest value
\begin{equation}
\begin{split}
\hat{\nu}_k[n]\hspace{-0.02cm}&=\mathds{E}\left\{\frac{I_k[N]R_k[N]}{w_0N+ D}\right\}\\ 
&=\mathds{E}\left\{\frac{I_k[N]\hat{R}_k[N]}{w_0N+ D} + \frac{I_k[N]\Psi_k[N]}{w_0N+ D} \right\}\\
&=\mathds{E}\left\{\frac{I_k[N]\hat{R}_k[N]}{w_0N+ D}\right\}.
\end{split}
\label{equ:schedulingMetric2}
\end{equation}
The third line in (\ref{equ:schedulingMetric2}) is obtained, since the expectation w.r.t. the discrete random variable can be split up into a finite sum, where each summand is equal to zero, due to $\mathds{E}\{\Psi_k[n]\}=0$. The expectation in the third line of (\ref{equ:schedulingMetric2}) can also be calculated by summing over terms for the discrete events, weighted by the probability of occurrence, to
\begin{equation}
\hat{\nu}_k[n] = I_k[N]\hat{R}_k[N]\sum_{m=1}^M\frac{\mathds{P}\{D=d_m\}}{w_0N+ d_m}.
\label{equ:finalMetric}
\end{equation}
The expected throughput is given as $\tilde{T}_k[N] = 1/\sum_{m=1}^M\frac{\mathds{P}\{D=d_m\}}{w_0N+ d_m}$. 
As for immediate feedback, (\ref{equ:finalMetric}) also maximized the utility function $u(\mathds{E}\{\f{T}[N])$ for the delayed feedback case.
\end{proof}

\section{Simulation Results}
\label{sec:simulationResults}
In this section, we illustrate the efficiency of the derived scheduling schemes. We evaluate the performance based on Monte-Carlo simulations 
considering a simple scenario, where a single BS serves two users on orthogonal resources (one in each time slot). Note that two users refer to the simplest scenario although the basic effects are captured.
The users are randomly placed within a circular area around the BS with radius $\unit[250]{m}$ which corresponds to an inter-site distance of $\unit[500]{m}$. The respective mean channel gain results from $\lambda_k=\beta d_k^{-\alpha}$, where $\alpha=3.5$ and $\beta=10e^{-14.45}$ is chosen according to 3GPP urban macro scenario \cite{3GPP10}.  
The simulations are done for 10,000 user drops, where for each drop 100 subsequent channel realizations are generated. The receive SNR at the border of the serving area (at $\unit[250]{m}$ distance from the BS) is assumed to be either $\unit[5]{dB}$ or $\unit[10]{dB}$. The impairment of CSI is generated by assuming a normalized coherence time of $T_{\T{C}} = 10$ transmission blocks. In an LTE system with a transmission block duration of 1 ms and a carrier frequency of 2.6 GHz, this refers to user mobilities of 10 km/h. Furthermore, we assume $N_{\T{P}}=8$ pilot signals are included per block and no quantization error is assumed ($Q = \infty$). The channel uncertainty (MSE between actual channel and its estimate) as a function of the feedback/backhaul delay is illustrated in \figref{fig:channeluncertainty}.
\begin{figure}[t]
	\centering
	\includegraphics[width=0.48\textwidth]{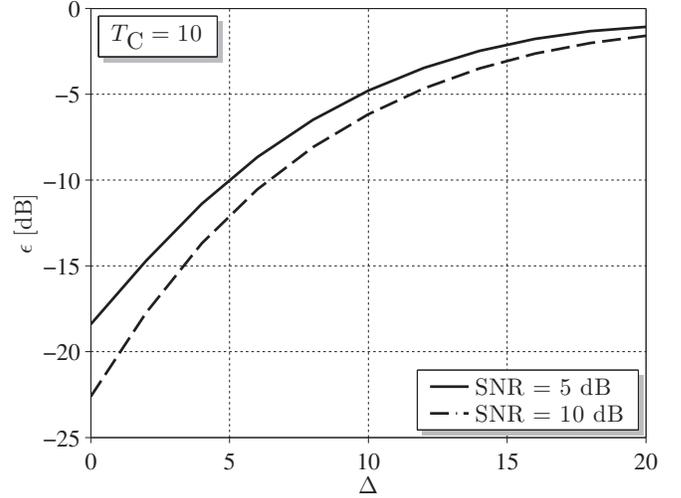}
	\caption{Channel uncertainty (MSE between actual channel and its estimate) as a function of the normalized feedback/backhaul delay for two different SNR values.}
	\label{fig:channeluncertainty}
\end{figure}
\begin{figure}[t]
	\centering
	\includegraphics[width=0.48\textwidth]{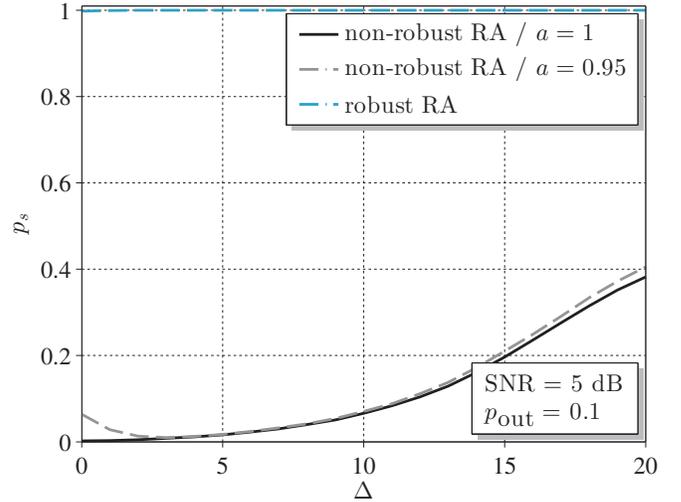}
	\caption{Portion of fulfilled target outage constraints as a function of the normalized delay.}
	\label{fig:percerntage}
\end{figure}
\figref{fig:percerntage} shows the portion of fulfilled target outage constraints (max. variance of 10 \%). If the proposed robust rate allocation (RA) is applied, the outage constraints are always satisfied. However, the non-robust RA with and without backoff shows a significant gap. The slight increase for large delays is caused by channel estimation/prediction, which inherently assesses a higher channel uncertainty with smaller CSI amplitudes and reduces the outage probability.
\par
For each channel realization, at first rate adaptation is performed for each user and then one of the users is selected for transmission, according to the respective scheduling scheme. In the following, we compare proportional fair (PF) scheduling with perfect CSI (P-CSI) and imperfect CSI (I-CSI), while for the latter case non-robust scheduling as well as the proposed robust scheduling is applied. The non-robust scheduler expects the available CSI to be perfect. Again, a back-off factor of $a=0.95$ and $a=1$ is used for simulations.
\par
\begin{figure}[t]
	\centering
	\includegraphics[width=0.48\textwidth]{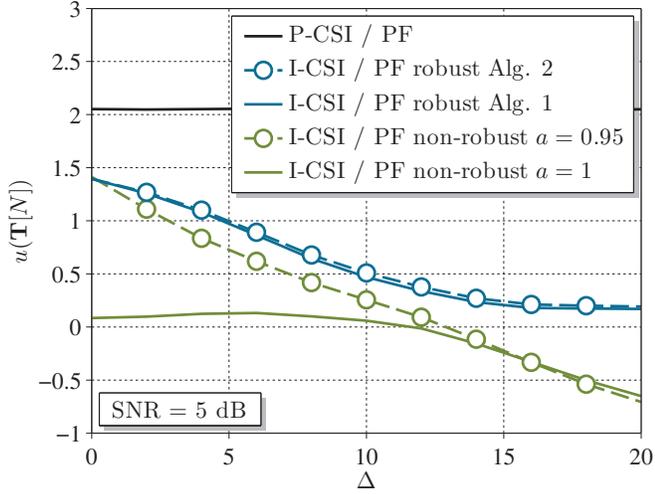}
	\caption{Proportional fair (PF) metric as a function of the normalized delay for $K=2$ users and SNR $=5$ dB.}
	\label{fig:scheduling_01}
\end{figure}
\begin{figure}[t]
	\centering
	\includegraphics[width=0.48\textwidth]{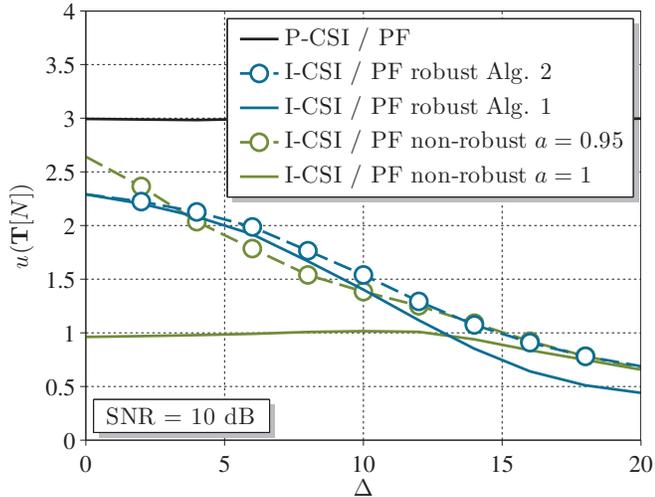}
	\caption{Proportional fair (PF) metric as a function of the normalized delay for $K=2$ users and SNR $=10$ dB.}
	\label{fig:scheduling_01b}
\end{figure}
For robust scheduling the two variants proposed in Section \ref{sec:immediateCSI} and Section \ref{sec:delayedCSI} are distinguished. The algorithm, which assumes immediate acknowledgements refers to \textit{Alg. 1}. In this case, the scheduler just uses the acknowledgements available. For \textit{Alg. 2} the algorithm for delay feedback is applied, where an inherent estimation of the missing transmission acknowledgements is performed.
\par
\begin{figure}[t]
	\centering
	\includegraphics[width=0.48\textwidth]{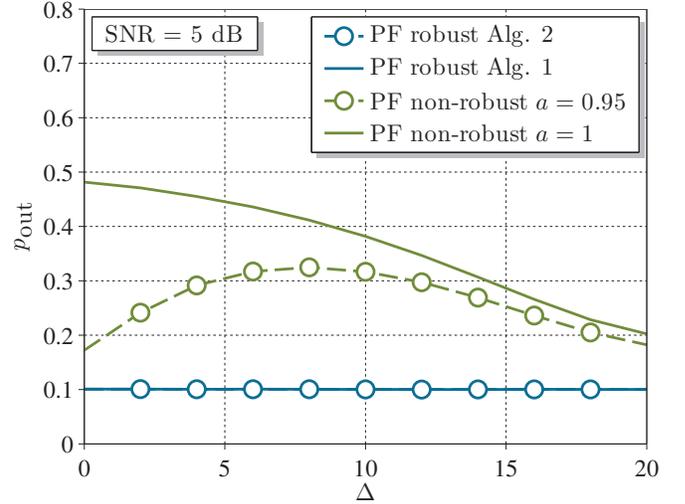}
	\caption{Resulting outage probability as a function of the normalized delay for $K=2$ users and SNR $=5$ dB.}
	\label{fig:scheduling_02}
\end{figure}
The PF metric as it is defined in (\ref{equ:utility}) as a function of the normalized delay $\Delta$ is illustrated in \figref{fig:scheduling_01}. Note that the PF utility metric $u(\f{T}[N])=\sum_{k=1}^K\log\left(T_k[N]\right)$ is defined as the sum of logarithmic throughput values. The highest values are achieved for perfect CSI (black solid line), which is independent of the delay. Performing the same algorithm with imperfect CSI (without backing off) results in the green solid line. This scheme basically suffers from high outages (see \figref{fig:scheduling_02}). Increasing the delay can even lead to marginal improvements in the PF metric. This is again a result from channel estimation/prediction as discussed before. Note that even for a delay of $\Delta=0$ blocks, the CSI is imperfect due to noisy pilot reception. Employing a back-off factor of $a=0.95$ (green dashed line with 'o' markers) causes a much higher percentage of successful transmissions, especially for small delays. Therefore, the PF metric significantly increases. However, even if the outage probability is very low in this region, it is not constant for all CSI amplitudes (see \figref{fig:resOutage}). Hence, the QoS constraints cannot be achieved at each transmission. 
\par
\begin{figure}[t]
	\centering
	\includegraphics[width=0.48\textwidth]{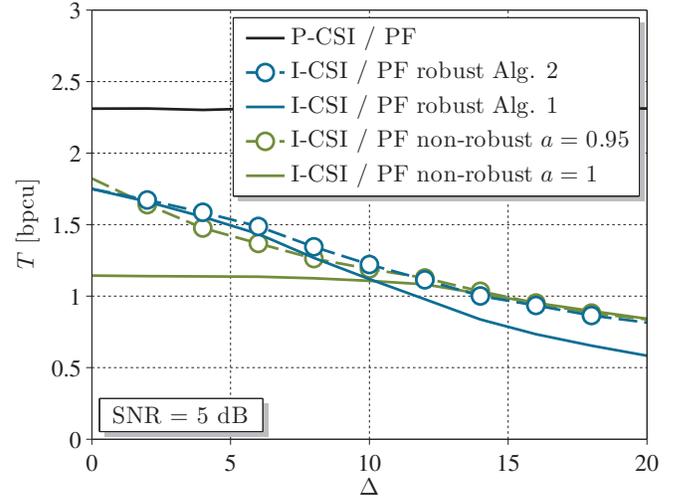}
	\caption{Mean user throughput as a function of the normalized delay for $K=2$ users and SNR $=5$ dB.}
	\label{fig:scheduling_03}
\end{figure}
The proposed robust PF algorithms (blue lines) outperform the non robust solutions for the SNR $=5$ dB. However, the performance advantage disappears by increasing the SNR (compare \figref{fig:scheduling_01} and \figref{fig:scheduling_01b}). For an SNR of $10$ dB, small performance degradations can be observed for low ($<4$ blocks) and high latency ($>10$ blocks). However, as can be seen from \figref{fig:scheduling_02}, the target outage probability of $p_{\T{out}}=0.1$ is achieved over the whole latency range and is particularly ensured at each transmission (see \figref{fig:resOutage}). Employing Alg. 2 (blue dashed line with 'o' markers) results in small gains of the PF metric compared to Alg. 1 (blue solid line). This gain increases with $\Delta$. 
\par
In addition, the mean user throughput is illustrated as a function of the normalized delay (see \figref{fig:scheduling_03}). In general, the throughput decreases with the delay. Small differences between the illustrated scheme can be observed, especially at delays between 5 and 15 blocks. For large delays, the non-robust scheme wth back-off factor performs equally to the robust ones, while only the robust algorithm ensures the QoS constrains. In addition, the performance gain by employing Alg. 2 instead of Alg. 1 is rather marginal but increases with the feedback delay.

\section{Conclusions}
\label{sec:conclusions}
In this work, a robust rate adaptation scheme together with a robust proportional fair (PF) scheduling algorithm have been presented, where both take into account that the available CSI is impaired. The algorithms aim to achieve a fixed outage probability at each transmission (independent of the channel amplitude), as it is of interest for delay critical applications. The robust PF algorithm has been derived for immediate  as well as for delayed transmission acknowledgements. For the first case, the robust PF scheduler only requires a change in the rate adaptation, while the scheduler itself can remain unchanged. If the success of former transmissions is not immediately known at the BS side, the scheduler inherently performs an estimation of the already achieved throughput per user. Performance advantages of the proposed schemes have been shown for a two user example scenario. It was shown that PF utility metric as well as throughput gains increase with latency. System level aspects for scenarios with more than two users will be addressed in our future work.

\appendices
\section{Algorithmic Solution for Robust Rate Adaptation}
\label{sec:appendix}

The rate adaptation function depends on three system parameters, the SNR $\gamma$, the channel uncertainty $\epsilon$ as well as the target outage probability $\bar{p}_{\T{out}}$. 
The numerical search for $\bar{R}$ based on (\ref{equ:outageProb}) with the given system parameters and estimated channel amplitude $\hat{g}$ is quite complex. Hence, a real-time calculation might be impractical. An alternative solution is the pre-calculation of the rate adaptation function for a certain channel amplitude resolution. In this case, the real-time adaptation refers to a look-up table search. Such approximative scheme refers to discrete rate adaptation as discussed, e.g., in \cite{GM13}.
\par 
In the following, the calculation of (\ref{equ:outage_last}) is discussed by taking numerical limitations into account, such as, floating point resolution. First, it can be observed that the argument of the modified Bessel function in (\ref{equ:outage_last}) decreases with the error $\epsilon$. Hence, small error values can result in large arguments and lead to values above the numerical resolution. 
In contrast, the exponential function in (\ref{equ:outage_last}) increases with the error $\epsilon$ and hence balances the product. In other words, the product of the two factors might fall into the numerical resolution, while each factor may not.
\par 
In order to overcome this problem, the modified Bessel function can be approximated for large arguments. The logarithm of the Bessel function behaves close to linear for large arguments, i.\,e., $\T{ln}(J_m(x))=a_mx + b_m$ can be utilized:
\begin{equation}
	\tilde{J}_m(x) = 
	\begin{cases}
		J_m(x)& x \leq X_m \\
		\T{exp}(a_mx+b_m)& x > X_m.\\
	\end{cases}
	\label{equ:approxBessel}
\end{equation}
The constants can be obtained by selecting two points $X_1$ and $X_2$ and calculating the steepness according to
\begin{equation}
	a_m = \frac{\T{ln}(J_m(X_2)) - \T{ln}(J_m(X_1))}{X_2 - X_1},
\end{equation}
while the shift is given by
\begin{equation}
	b_m = \T{ln}(J_m(X_1)) - a_mX_1.
\end{equation}
With the approximation in (\ref{equ:approxBessel}), the outage probability given in (\ref{equ:outage_last}) results in
\begin{equation}
	\tilde{p}_{\T{out}} = 1 - \sum_{m=0}^{\infty}\tilde{d}_m,
\end{equation} 
with the approximated summand
\begin{equation}
	\tilde{d}_m = \T{exp}\left(c_0\right)\tilde{J}_m(c_3b),
	\label{equ:d_m}
\end{equation} 
including the coefficients $c_0 = c_1+c_2b^2+m(\T{ln}(c_4) - \T{ln}(b))$, $c_1= - \hat{g}^2/(\lambda\epsilon)$, $c_2 = -1/(\lambda\epsilon)$, $c_3 = 2\hat{g}/(\lambda\epsilon)$, $c_4=\hat{g}$ and $b = \sqrt{2^R-1/\rho}$. 
\par
In case the modified Bessel function reaches the numerical limit, (\ref{equ:d_m}) can be written as
\begin{equation}
	\tilde{d}_m = \T{exp}\left(c_0+a_mc_3b+b_m\right)\quad\T{if } c_3b>X_m.
	\label{equ:approximation}
\end{equation} 
This expression allows to perform the calculation within a numerical stable range. However, the approximation includes inaccuracies which may lead to approximation errors. Inaccuracies can be observed for small target outage probabilities. Consequently, the robust algorithm would assign higher rates and the resulting outage probability exceeds the target, as illustrated in \figref{fig:resOut_tarOut}. However, this approximation error can be avoided by forcing the rate-assignment function to be monotonically increasing.

\begin{figure}[t]
	\centering
	\includegraphics[width=0.48\textwidth]{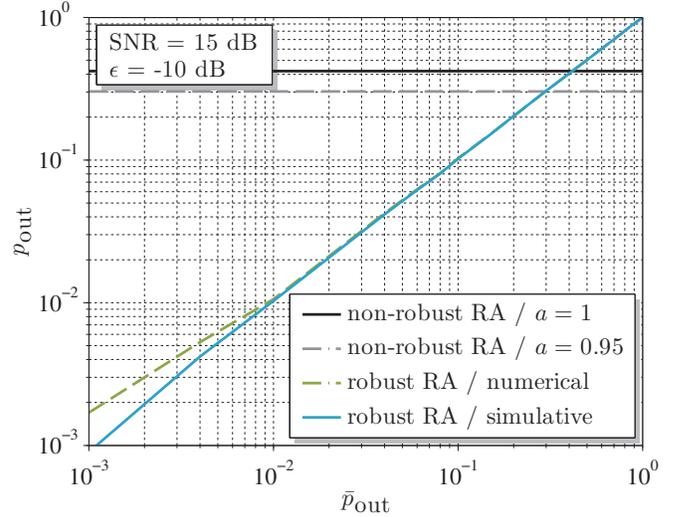}
	\caption{The resulting outage probability as a function of the targe outage probability.}
	\label{fig:resOut_tarOut}
	\vspace{-0.5cm}
\end{figure} 
In order to evaluate (\ref{equ:outage_last}) numerically, the sum over an infinite number of arguments need to be restricted to $M_{\T{max}}$, resulting in
\begin{equation}
	p_{\T{out}}\approx\tilde{p}_{\T{out}}=1 - \sum_{m=0}^{M_{\T{max}}}\tilde{d}_m.
	\label{equ:outage_numerical}
\end{equation}
Note that a small value of $M_{\T{max}}$ can cause further inaccuracies. However, simulative experiments showed that $M_{\T{max}}=150$ lead to a sufficient accuracy.

\section*{Acknowledgment}
The research leading to these results has received funding
from the European Union Seventh Framework Programme
(FP7/2007-2013) under grant agreement no. 317941. The
authors would like to acknowledge the contributions of their
colleagues in iJOIN, although the views expressed are those
of the authors and do not necessarily represent the project.

\bibliographystyle{IEEEtran}
\bibliography{references}

\begin{thebibliography}{10}
\providecommand{\url}[1]{#1}
\csname url@samestyle\endcsname
\providecommand{\newblock}{\relax}
\providecommand{\bibinfo}[2]{#2}
\providecommand{\BIBentrySTDinterwordspacing}{\spaceskip=0pt\relax}
\providecommand{\BIBentryALTinterwordstretchfactor}{4}
\providecommand{\BIBentryALTinterwordspacing}{\spaceskip=\fontdimen2\font plus
\BIBentryALTinterwordstretchfactor\fontdimen3\font minus
  \fontdimen4\font\relax}
\providecommand{\BIBforeignlanguage}[2]{{%
\expandafter\ifx\csname l@#1\endcsname\relax
\typeout{** WARNING: IEEEtran.bst: No hyphenation pattern has been}%
\typeout{** loaded for the language `#1'. Using the pattern for}%
\typeout{** the default language instead.}%
\else
\language=\csname l@#1\endcsname
\fi
#2}}
\providecommand{\BIBdecl}{\relax}
\BIBdecl

\bibitem{LHL08}
D.~Love, R.~Heath, V.~Lau, D.~Gesbert, B.~Rao, and M.~Andrews, ``{An Overview
  of Limited Feedback in Wireless Communication Systems},'' \emph{IEEE Journal
  on Selected Areas in Communications}, vol.~26, no.~8, pp. 1341--1365, 2008.

\bibitem{FOF13b}
R.~Fritzsche, E.~Ohlmer, and G.~Fettweis, ``{Where to Predict the Channel for
  Cooperative Multi-Cell Transmission over Correlated Subcarriers?}'' in
  \emph{Proc. IEEE International Conference on Acoustics, Speech, and Signal
  Processing (ICASSP '13)}, 2013.

\bibitem{TV08}
D.~Tse and P.~Viswanath, \emph{{Fundamentals of Wireless
  Communications}}.\hskip 1em plus 0.5em minus 0.4em\relax Cambridge University
  Press, 2008.

\bibitem{KH95}
R.~Knopp and P.~Humblet, ``{Information Capacity and Power Control in
  Single-Cell Multiuser Communications},'' in \emph{Proc. IEEE International
  Conference on Communications (ICC '95)}, 1995.

\bibitem{LG01}
L.~Li and A.~Goldsmith, ``{Capacity and Optimal Resource Allocation for Fading
  Broadcast Channels - PartI: Ergodic Capacity},'' \emph{IEEE Transactions on
  Information Theory}, vol.~47, no.~3, pp. 1083--1102, Mar 2001.

\bibitem{LBS99}
S.~Lu, V.~Bharghavan, and R.~Srikant, ``{Fair Scheduling in Wireless Packet
  Networks},'' \emph{IEEE/ACM Transactions on Networking}, vol.~7, no.~4, pp.
  473--489, Aug 1999.

\bibitem{Kel97}
F.~Kelly, ``{Charging and Rate Control for Elastic Traffic},'' \emph{European
  Transactions on Telecommunications}, 1997.

\bibitem{JPP00}
A.~Jalali, R.~Padovani, and R.~Pankaj, ``{Data Throughput of CDMA-HDR a High
  Efficiency-High Data Rate Personal Communication Wireless System},'' in
  \emph{Proc. IEEE Vehicular Technology Conference (VTC '00)}, 2000.

\bibitem{BW01}
S.~Borst and P.~Whiting, ``{Dynamic Rate Control Algorithms for HDR Throughput
  Optimization},'' in \emph{Proc. IEEE Computer and Communications Societies
  (INFOCOM '01)}, 2001.

\bibitem{CB07}
J.-G. Choi and S.~Bahk, ``{Cell-Throughput Analysis of the Proportional Fair
  Scheduler in the Single-Cell Environment},'' \emph{IEEE Transactions on
  Vehicular Technology}, vol.~56, no.~2, pp. 766--778, March 2007.

\bibitem{CC06}
B.~B. Chen and M.~C. Chan, ``{Proportional Fairness for Overlapping Cells in
  Wireless Networks},'' in \emph{Proc. IEEE Vehicular Technology Conference
  (VTC '06)}, 2006.

\bibitem{ZFL11}
H.~Zhou, P.~Fan, and J.~Li, ``{Global Proportional Fair Scheduling for Networks
  With Multiple Base Stations},'' \emph{IEEE Transactions on Vehicular
  Technology}, vol.~60, no.~4, pp. 1867--1879, May 2011.

\bibitem{PSL03}
T.~Park, O.-S. Shin, and K.-B. Lee, ``Proportional fair scheduling for wireless
  communication with multiple transmit and receive antennas,'' in \emph{Proc.
  IEEE Vehicular Technology Conference (VTC '03)}, 2003.

\bibitem{Lau05}
V.~Lau, ``{Proportional Fair Space-Time Scheduling for Wireless
  Communications},'' \emph{IEEE Transactions on Communications}, vol.~53,
  no.~8, pp. 1353--1360, Aug 2005.

\bibitem{LNZ10}
L.~Liu, Y.-H. Nam, and J.~Zhang, ``{Proportional Fair Scheduling for Multi-Cell
  Multi-User MIMO Systems},'' in \emph{44th Annual Conference on Information
  Sciences and Systems (CISS '10)}, March 2010, pp. 1--6.

\bibitem{WE09}
I.~Wong and B.~Evans, ``Optimal resource salocation in the ofdma downlink with
  imperfect channel knowledge,'' \emph{IEEE Transactions on Communications},
  vol.~57, no.~1, pp. 232--241, January 2009.

\bibitem{AAK+11}
R.~Aggarwal, M.~Assaad, C.~Koksal, and P.~Schniter, ``Joint scheduling and
  resource allocation in the ofdma downlink: Utility maximization under
  imperfect channel-state information,'' \emph{IEEE Transactions on Signal
  Processing}, vol.~59, no.~11, pp. 5589--5604, Nov 2011.

\bibitem{Ros12}
P.~Rost, ``{Achievable Net-Rates in Multi-User OFDMA with Partial CSI and
  Finite Channel Coherence},'' in \emph{Proc. IEEE Vehicular Technology
  Conference (VTC '12 Fall)}, 2012.

\bibitem{FRF14}
R.~Fritzsche, P.~Rost, and G.~Fettweis, ``{Robust Proportional Fair Scheduling
  with Imperfect CSI and Fixed Outage Probability},'' in \emph{Proc. IEEE
  Symposium on Personal, Indoor and Mobile Radio Communications (PIMRC '14)},
  2014.

\bibitem{SK76}
A.~R.~K. Sastry and L.~N. Kanal, ``{Hybrid Error Control Using Retransmission
  and Generalized Burst-Trapping Codes},'' \emph{IEEE Transactions on
  Communications}, vol.~24, no.~4, pp. 385--393, Apr 1976.

\bibitem{LMS07}
C.~Lott, O.~Milenkovic, and E.~Soljanin, ``{Hybrid ARQ: Theory, State of the
  Art and Future Directions},'' in \emph{Proc. IEEE Information Theory Workshop
  (ITW '07)}, 2007.

\bibitem{KAS00}
Y.-C. Ko, M.-S. Alouini, and M.~K. Simon, ``{Outage Probability of Diversity
  Systems over Generalized Fading Channels},'' \emph{IEEE Transactions on
  Communications}, vol.~48, no.~11, pp. 1783--1787, Nov 2000.

\bibitem{KCB+11}
S.~M. Kim, W.~Choi, T.-W. Ban, and D.~K. Sung, ``Optimal rate adaptation for
  hybrid arq in time-correlated rayleigh fading channels,'' \emph{IEEE
  Transactions on Wireless Communications}, vol.~10, no.~3, pp. 968--979, March
  2011.

\bibitem{WJ10}
P.~Wu and N.~Jindal, ``{Performance of Hybrid-ARQ in Block-Fading Channels: A
  Fixed Outage Probability Analysis},'' \emph{IEEE Transactions on
  Communications}, vol.~58, no.~4, pp. 1129--1141, April 2010.

\bibitem{GM13}
S.~Guharoy and N.~Mehta, ``{Joint Evaluation of Channel Feedback Schemes, Rate
  Adaptation, and Scheduling in OFDMA Downlinks With Feedback Delays},''
  \emph{Vehicular Technology, IEEE Transactions on}, vol.~62, no.~4, pp.
  1719--1731, May 2013.

\bibitem{RBD+14}
P.~Rost, C.~J. Bernardos, A.~De~Domenico, M.~Di~Girolamo, M.~Lalam, A.~Maeder,
  D.~Sabella, and D.~Wübben, ``{Cloud Technologies for Flexible 5G Radio Access
  Networks},'' \emph{IEEE Communications Magazine}, accepted for publication
  2014.

\bibitem{FF11}
R.~Fritzsche and G.~Fettweis, ``{CSI Distribution for Joint Processing in
  Cooperative Cellular Networks},'' in \emph{Proc. IEEE Vehicular Technology
  Conference (VTC '11-Fall)}, 2011.

\bibitem{DF11}
F.~Diehm and G.~Fettweis, ``{Centralized Scheduling for Joint Decoding
  Cooperative Networks Subject to Signalling Delays},'' in \emph{Proc. IEEE
  Vehicular Technology Conference (VTC '11 Fall)}, 2011.

\bibitem{FOF13}
R.~Fritzsche, E.~Ohlmer, and G.~Fettweis, ``{Where to Predict the Channel in
  Cooperative Cellular Networks with Backhaul Delay?}'' in \emph{Proc. IEEE
  International Conference on Systems, Communications and Coding (SCC '13)},
  2013.

\bibitem{Kay93}
S.~M. Kay, \emph{{Fundamentals of Signal Processing: Estimation Theory}},
  1st~ed.\hskip 1em plus 0.5em minus 0.4em\relax Prentice Hall PTR, 1993.

\bibitem{Nut75}
A.~H. Nuttall, ``Some integrals involving the q function,'' \emph{IEEE
  Transactions on Information Theory}, vol.~21, no.~1, pp. 95--96, Jan 1975.

\bibitem{3GPP10}
{Technical Specification Group Radio Access Network}, ``{3GPP TR 36.814},''
  Release 9 2010.

\end{thebibliography}

\end{document}